\documentclass{jpsj3}
\usepackage{txfonts}

\usepackage{amsmath}
\usepackage{theorem}
\usepackage{amssymb}
\usepackage{ascmac}
\usepackage{braket}
\usepackage{delarray}
\usepackage{here}

\newtheorem{theorem}{Theorem}[section]
\newtheorem{proof}{Proof.}

\newtheorem{corollary}{Corollary}[section]
\usepackage{latexsym}
\def\qed{\hfill $\Box$}

\usepackage{epsfig}
\usepackage{bm}
\usepackage{color}
\usepackage{ulem}

\newcommand{\be}{\begin{eqnarray}}
\newcommand{\ee}{\end{eqnarray}}
\newcommand{\ba}{\begin{array}}
\newcommand{\ea}{\end{array}}
\newcommand{\bmat}{\left(\begin{array}}
\newcommand{\emat}{\end{array}\right)}
\newcommand{\no}{\nonumber}

\title{Inequality for local energy of Ising model with quenched randomness and its application}

\author{ Manaka Okuyama$^1$\thanks{manaka.okuyama.d2@tohoku.ac.jp} and Masayuki Ohzeki$^{1,2,3}$}
\inst{$^1$Graduate School of Information Sciences, Tohoku University, Sendai 980-8579, Japan \\
$^2$Institute of Innovative Research, Tokyo Institute of Technology, Yokohama 226-8503, Japan \\
$^3$Sigma-i Co. Ltd., Tokyo 108-0075, Japan} %\\

\abst{In this study, we extend the lower bound on the average of the local energy of the Ising model with quenched randomness [J. Phys. Soc. Jpn. {\bf76,} 074711 (2007)] obtained for a symmetric distribution to an asymmetric one.
Compared with the case of symmetric distribution, our bound has a non-trivial term.
By applying the acquired bound to a Gaussian distribution, we obtain the lower bounds on the expectation of the square of the correlation function.
Thus, we demonstrate that in the Ising model in a Gaussian random field, the spin-glass order parameter generally has a finite value at any temperature, regardless of the forms of the other interactions.}

\begin{document}
\maketitle

%%%%%%%%%%%%%%%%%%%%%%%%%%%%%%%%%%%%%%%%%%%%%%%%%%%%%%%%%%%%%%%%%
%%%%%%%%%%%%%%%%%%%%%%%%%%%%%%%%%%%%%%%%%%%%%%%%%%%%%%%%%%%%%%%%%
\section{Introduction}

Spin-glass models describe magnetic materials that randomly interact spatially.
The mean field theory of spin-glass models, e.g., the Sherrington--Kirkpatrick model, has been solved rigorously by the full replica symmetry breaking solution~\cite{Parisi,Parisi2,Guerra,Talagrand}; however, it is extremely difficult to obtain analytical results for finite-dimensional models, except on the Nishimori line~\cite{Nishimori}.
Although analytical approaches ~\cite{ON} for two-dimensional systems are slightly progressing, those for three-dimensional systems have been primarily neglected, except numerical analysis.

In ferromagnetic spin models, correlation inequalities play an important role in non-perturbative analysis and yield rigorous results for unsolvable models.
Correlation inequalities are also valid for the Ising model in a random field.
A recent study~\cite{KTZ} proved based on the Fortuin--Kasteleyn--Ginibre inequality that the random-field Ising model comprising two-body interactions for all the lattice and field distributions does not have a spin-glass phase. 
Therefore, it is expected that the concept of correlation inequalities will be important for a rigorous analysis of spin-glass models, and their establishment for spin-glass models is a very important problem.

Some previous studies have been conducted on the correlation inequalities in spin-glass models.
A recent study~\cite{CG,CL} exhibited that the response of the quenched average of a partition function with respect to the variance is generally positive, which is considered as the counterpart of the Griffiths first inequality in spin-glass models.
In addition, for various bond randomness, including Gaussian and binary distribution types, it is shown that the counterpart of the Griffiths second inequality holds on the Nishimori line~\cite{MNC, Kitatani}.
However, correlation inequalities, as in the case of ferromagnetic spin models, have not been obtained in general, and a rigorous analysis based on them is yet to be conducted satisfactorily for spin-glass models.

In this study, we obtain a lower bound on the quenched average of the local energy for the Ising model with quenched randomness.
The result of a previous study~\cite{KNA} that was limited to a symmetric distribution is generalized to an asymmetric distribution.
Furthermore, as a simple application of the acquired inequality, we obtain the correlation inequalities for a Gaussian distribution.
We demonstrate that the expectation of the square of the correlation function generally has a finite lower bound at any temperature.
Thus, we prove that the spin-glass order parameter has a finite lower bound in the Ising model in a Gaussian random field, regardless of the forms of the other interactions.

The organization of this paper is as follows.
In Sec. II, we define the model and present the method to obtain the lower bound on the average of the local energy for the Ising model with quenched randomness.
In Sec. III, we describe the application of the acquired inequality when the randomness of the interactions follows a Gaussian distribution.
Finally, our conclusion is presented in Sec. IV.

%%%%%%%%%%%%%%%%%%%%%%%%%%%%%%%%%%%%%%%%%%%%%%%%%%%%
%%%%%%%%%%%%%%%%%%%%%%%%%%%%%%%%%%%%%%%%%%%%%%%%%%%%
\section{Lower bound on local energy for asymmetric distribution of randomness}
Following Ref. \cite{KNA}, we consider a generic form of the Ising model,
\be
H&=&- \sum_{B \subset{V}} J_{B} \sigma_B ,
\\
\sigma_B&\equiv& \prod_{i\in B} \sigma_i ,
\ee
where $V$ is the set of sites, the sum over $B$ is over all the subsets of $V$ in which interactions exist, and the lattice structure adopts any form.
The probability distribution of a random interaction $J_B$ is represented as $P_B (J_B)$.
The probability distributions can be generally different from each other, i.e., $P_B (x)\neq P_{C} (x)$, and are also allowed to present no randomness, i.e., $P_B(J_B)=\delta(J-J_B)$.

The correlation function for a set of fixed interactions, $\{J_{B}\}$, is expressed as
\be
\langle\sigma_A \rangle_{\{J_{B}\}}&=& \frac{ \Tr \sigma_A \exp\left(\beta \sum_{B \subset{V}} J_{B} \sigma_B \right)}{  \Tr \exp\left(\beta \sum_{B \subset{V}} J_{B} \sigma_B \right)} .
\ee
The configurational average over the distribution of the randomness of the interactions is written as 
\be
\mathbb{E}\left[g(\{J_{B}\})  \right] =  \left(\prod_{B \subset{V}} \int_{-\infty}^\infty dJ_B P_B(J_B) \right) g(\{J_{B}\}).
\ee
For example, the quenched average of the correlation function is obtained as
\be
\mathbb{E}\left[ \langle\sigma_A \rangle_{\{J_{B}\}} \right]&=&  \left(\prod_{B \subset{V}} \int_{-\infty}^\infty dJ_B P_B(J_B) \right)  \frac{ \Tr \sigma_A \exp\left(\beta \sum_{B \subset{V}} J_{B} \sigma_B \right)}{  \Tr \exp\left(\beta \sum_{B \subset{V}} J_{B} \sigma_B \right)} .
\ee
Our result is the following theorem.
%%%%%%%%%%%%%%%%%%%%%%%%%%%%%%%%%%%%%
\begin{theorem}\label{th1}
When the distribution function of the randomness satisfies 
\be
P_A(-J_A)&=& \exp(-2\beta_{\mathrm NL} J_A) P_A(J_A),
\ee
then for any even function $f(J_A)\ge0$, the system defined above satisfies the following inequality:
\be
\mathbb{E}\left[-J_{A} f(J_{A}) \langle\sigma_A \rangle_{\{J_{B}\}}  \right] &\ge& \mathbb{E}\left[-J_{A} f(J_{A})\tanh(\beta J_{A}) -J_{A}f(J_{A})(1-e^{-\beta_{\mathrm NL}J_{A}}) \frac{1}{\sinh(2\beta J_{A})}  \right] \label{inequality-2} .
\ee
\end{theorem}
%%%%%%%%%%%%%%%%%%%%%%%%%%%%%%%%%%%%%
We  note that the right-hand side of Eq. (\ref{inequality-2}) does not depend on the other interactions.
When the distribution function is symmetric, i.e., $P_A(-J_A)=P_A(J_A)$ ($\beta_{\mathrm NL}=0$) and $f(J_A)=1$, Eq. (\ref{inequality-2}) is reduced to 
\be
\mathbb{E}\left[-J_{A} \langle\sigma_A \rangle_{\{J_{B}\}}  \right] &\ge& \mathbb{E}\left[-J_{A} \tanh(\beta J_{A}) \right] , \label{pre-ieq}
\ee
which is in accordance with the result in Ref. \cite{KNA}. In this case, an intuitive explanation of the inequality is possible: the local energy is generally larger than or equal to the energy in the absence of all the other interactions.
However, for $\beta_{\mathrm NL}\neq0$, it is difficult to provide an intuitive explanation, because the second term in the right-hand side of Eq. (\ref{inequality-2}) does not have a physical relevance.
On the other hand, the right-hand side of Eq. (\ref{inequality-2}) can be rewritten as
\be
&&\mathbb{E}\left[-J_{A} f(J_{A})\tanh(\beta J_{A}) -J_{A}f(J_{A})(1-e^{-\beta_{\mathrm NL}J_{A}}) \frac{1}{\sinh(2\beta J_{A})}  \right]
\no\\
&=&\mathbb{E}\left[-J_{A}f(J_{A})\tanh(\beta J_{A}) \right]  - \int_{0}^\infty dJ_A P_A(J_A) J_{A} f(J_A)  e^{-2\beta_{\mathrm NL} J_A} (1-e^{\beta_{\mathrm NL}J_{A}})^2 \frac{1}{\sinh(2\beta J_{A})}  
\no\\
&\le&\mathbb{E}\left[-J_{A}f(J_{A})\tanh(\beta J_{A}) \right],
\ee
which suggests that the local energy can be lower than the energy in the absence of all the other interactions, unlike in the case of $\beta_{\mathrm NL}=0$.
We also note that the second term in the right-hand side of Eq. (\ref{inequality-2}) is numerically very small.
Thus, to establish the bound for $\beta_{\mathrm NL}\neq0$, such a small correction must considered.

%%%%%%%%%%%%%%%%%%%%%%%%%%%%%%%%%%%%%
%%%%%%%%%%%%%%%%%%%%%%%%%%%%%%%%%%%%%
\begin{proof}
{\rm
We define $Z(\beta, J_{A})$ and $\langle\sigma_A \rangle_{J_{A}}$ as
\be
Z(\beta, J_{A}) &=&  \sum_{\{ \sigma \}} \exp\left(\beta \sum_{B \subset{V}\setminus A} J_{B} \sigma_B +\beta  J_{A} \sigma_A  \right),
\\
\langle\sigma_A \rangle_{J_{A}}  &=& \frac{\sum_{\{ \sigma \}} \sigma_A \exp\left(\beta \sum_{B \subset{V}\setminus A} J_{B} \sigma_B +\beta  J_{A} \sigma_A  \right)}{\sum_{\{ \sigma \}} \exp\left(\beta \sum_{B \subset{V}\setminus A} J_{B} \sigma_B +\beta  J_{A} \sigma_A  \right)} .
\ee
We note that  $\langle\sigma_A \rangle_{J_{A}}=\langle\sigma_A \rangle_{\{J_{B}\}}$ but $\langle\sigma_A \rangle_{-J_{A}} \neq \langle\sigma_A \rangle_{\{J_{B}\}}$.
Subsequently, we obtain
\be
\frac{Z(\beta, J_{A})}{Z(\beta,-J_{A})} &=&\cosh(2\beta J_{A}) +  \langle\sigma_A \rangle_{-J_{A}}  \sinh(2\beta J_{A})
\no\\
&=&e^{\beta_{\mathrm NL}J_{A}}+  \Gamma(\beta,-J_{A})  \frac{\sinh(2\beta J_{A})}{J_{A}}\ge 0 \label{part-relation-2-1},
\\
\frac{Z(\beta,-J_{A})}{Z(\beta,J_{A})} &=&\cosh(2\beta J_{A}) -  \langle\sigma_A \rangle_{J_{A}}  \sinh(2\beta J_{A})
\no\\
&=&e^{-\beta_{\mathrm NL}J_{A}}+  \Gamma(\beta, J_{A}) \frac{\sinh(2\beta J_{A})}{J_{A}}\ge 0 \label{part-relation-2-2},
\ee
where $\Gamma(\beta,J_{A})$ is defined as 
\be
\Gamma(\beta,J_{A}) \equiv -J_{A} \langle \sigma_A \rangle_{J_{A}} +J_{A} \tanh(\beta J_{A}) +(1 -e^{-\beta_{\mathrm NL}J_{A}} )\frac{J_{A}}{\sinh(2\beta J_{A})} .\label{def-Gam}
\ee
Since Eq. (\ref{part-relation-2-1}) is the reciprocal of  Eq. (\ref{part-relation-2-2}), we obtain
\be
e^{-2\beta_{\mathrm NL}J_{A}}\Gamma(\beta,-J_{A})  &=&\frac{-e^{-\beta_{\mathrm NL}J_{A}}\Gamma(\beta, J_{A})  }{e^{-\beta_{\mathrm NL}J_{A}}+  \Gamma(\beta, J_{A}) \frac{\sinh(2\beta J_{A})}{J_{A}}} \label{inverse-relation-2} .
\ee
Then, from Eq. (\ref{def-Gam}), we immediately obtain  
\be
&&\mathbb{E}\left[-J_{A} f(J_A)\langle\sigma_A \rangle_{\{J_{B}\}}  \right] 
\no\\
&=& \mathbb{E}\left[f(J_A)\Gamma(\beta, J_{A})-J_{A}f(J_A) \tanh(\beta J_{A}) -J_{A}f(J_A)(1-e^{-\beta_{\mathrm NL}J_{A}} ) \frac{1}{\sinh(2\beta J_{A})}  \right] \label{def-st} .
\ee
Furthermore, for any even function $f(J_A)\ge0$, we obtain $\mathbb{E}\left[ f(J_A) \Gamma(\beta, J_{A}) \right]\ge0$, because
\be
\mathbb{E}\left[ f(J_A)\Gamma(\beta, J_{A}) \right]
&=& \int_{-\infty}^\infty dJ_A P_A(J_A) f(J_A) \mathbb{E}\left[  \Gamma(\beta, J_{A}) \right]'
\no\\
&=& \int_{0}^\infty dJ_A P_A(J_A)  f(J_A) \mathbb{E}\left[ \Gamma(\beta, J_{A}) +  \exp(-2\beta_{\mathrm NL}J_A) \Gamma(\beta, -J_{A}) \right]'
\no\\
&=& \int_{0}^\infty dJ_A P_A(J_A) f(J_A)  \mathbb{E}\left[\frac{ \Gamma^2(\beta, J_{A}) \frac{\sinh(2\beta J_{A})}{J_{A}}}{e^{-\beta_{\mathrm NL}J_{A}}+  \Gamma(\beta, J_{A}) \frac{\sinh(2\beta J_{A})}{J_{A}}} \right]'
\no\\
&\ge&0 , \label{posiive2} 
\ee
where $\mathbb{E}[\cdots]'$ denotes the configurational average over the randomness of the interactions other than $J_A$. We used Eq. (\ref{inverse-relation-2}) in the third identity and Eq. (\ref{part-relation-2-2}) in the last inequality. 
Thus, Eqs. (\ref{def-st}) and (\ref{posiive2}) yield  Eq. (\ref{inequality-2}). 
}
\qed\end{proof}

%%%%%%%%%%%%%%%%%%%%%%%%%%%%%%%%%%%%%%%%%%%%%%%%%%%%
%%%%%%%%%%%%%%%%%%%%%%%%%%%%%%%%%%%%%%%%%%%%%%%%%%%%
\section{Application to Gaussian spin-glass model}
In this section, we present the application of Eq. (\ref{inequality-2}) to a spin-glass model with a Gaussian distribution.
We note that the result for $\beta_{\mathrm NL}=0$ in Ref. \cite{KNA} is sufficient to obtain the inequalities that are presented in this section.

First, we consider the distinct case, $P_A(J_{0,A}-J_A)= P_A(J_{0,A}+J_A)$.
Then, we obtain the following result:
%%%%%%%%%%%%%%%%%%%%%%%%%%%%%%%%%%%%%
\begin{corollary}
When the distribution function of the randomness satisfies 
\be
P_A(J_{0,A}-J_A)&=& P_A(J_{0,A}+J_A),
\ee
then for any even function $f(J_A)\ge0$, the system defined above satisfies the following inequality:
\be
\mathbb{E}\left[\left(J_{0,A}-J_{A}\right) f(J_{A}-J_{0,A}) \langle\sigma_A \rangle_{\{J_{B}\}}  \right] &\ge& \mathbb{E}\left[\left(J_{0,A}-J_{A}\right)f(J_{A}-J_{0,A}) \tanh(J_{A}-J_{0,A})   \right] \label{inequality-1} .
\ee
\end{corollary}
%%%%%%%%%%%%%%%%%%%%%%%%%%%%%%%%%%%%%
%%%%%%%%%%%%%%%%%%%%%%%%%%%%%%%%%%%%%
\begin{proof}{}
{\rm
We regard $P_A(J_{0,A}+J_A)$ as a new probability distribution $P_A'(J_A)$, where $P_A'(J_A)$ is symmetric.
Therefore, using Eq. (\ref{inequality-2}) for $\beta_{\mathrm NL}=0$, we obtain 
\be
\mathbb{E}\left[\left(J_{0,A}-J_{A}\right) f(J_{A}-J_{0,A}) \langle\sigma_A \rangle_{\{J_{B}\}}  \right] &=& \int_{-\infty}^\infty dJ_A P_A(J_{0,A}+J_A) \mathbb{E}\left[-J_{A}f(J_{A}) \langle\sigma_A \rangle_{J_A+J_{0,A}}  \right]' 
\no\\
&\ge&\int_{-\infty}^\infty dJ_A P_A(J_{0,A}+J_A)  \left( -J_{A} \right) f(J_A) \tanh(J_{A})  
\no\\
&=&\mathbb{E}\left[\left(J_{0,A}-J_{A}\right)f(J_{A}-J_{0,A}) \tanh(J_{A}-J_{0,A})   \right] .
\ee
}
\qed\end{proof}
In the following, using Eq. (\ref{inequality-1}), we obtain several inequalities.

%%%%%%%%%%%%%%%%%%%%%%%%%%%%%%%%%%%%%%%%%%%%%%%%%%%%
%%%%%%%%%%%%%%%%%%%%%%%%%%%%%%%%%%%%%%%%%%%%%%%%%%%%
\subsection{Correlation inequality for Gaussian spin-glass model}
Next, we consider the case where all the interactions follow a Gaussian distribution with mean $J_{0,B}$ and variance $\Lambda_{B}^2$.
All the $J_{0,B}$ and $\Lambda_{B}^2$ can adopt different values.
We denote the configurational average over the distribution of the randomness of the interactions as $\mathbb{E}\left[\cdots  \right]_{\left\{J_{0,B},\Lambda_{B}^2 \right\}}$.
Then, we obtain the following result:
%%%%%%%%%%%%%%%%%%%%%%%%%%%%%%%%%%%%%
\begin{corollary}
For the quenched average of the square of the correlation function, we obtain a lower bound, 
 \be
 \mathbb{E}\left[  \tanh^2(\beta  J_{A}) \right]_{\left\{0,\Lambda_{A}^2 \right\}} \le  \mathbb{E}\left[ \langle\sigma_A \rangle_{\{J_{B} \}}^2  \right]_{\left\{J_{0,B},\Lambda_{B}^2 \right\}} \label{corr-ineq} .
\ee
\end{corollary}
%%%%%%%%%%%%%%%%%%%%%%%%%%%%%%%%%%%%%%%%%%%%%%%%%%%%
We note that the left-hand side of Eq. (\ref{corr-ineq}) is independent of mean  $\{J_{0,b}\}$.
Inequality (\ref{corr-ineq}) indicates that the expectation of the square of the correlation function is generally a finite non-zero value, regardless of the other interactions.
This behavior is quite different from those of the correlation functions of ferromagnetic models and may reflect the fact that the counterpart of the Griffiths second inequality has not been established in spin-glass models~\cite{CUV}.

%%%%%%%%%%%%%%%%%%%%%%%%%%%%%%%%%%%%%
%%%%%%%%%%%%%%%%%%%%%%%%%%%%%%%%%%%%%
%%%%%%%%%%%%%%%%%%%%%%%%%%%%%%%%%%%%%
\begin{proof}{}
{\rm
For the Gaussian distribution with mean $J_{0,B}$ and variance $\Lambda_{B}^2$, and $f(J_{A})=1$, Eq. (\ref{inequality-1}) is reduced to
\be
\mathbb{E}\left[(J_{0,A}-J_{A})\langle\sigma_A \rangle_{\{ J_{B}\}}  \right]_{\left\{J_{0,B},\Lambda_{B}^2 \right\}}&=&\mathbb{E}\left[ -J_{A} \langle\sigma_A \rangle_{\{J_{B}+J_{0,B}\}}  \right]_{\left\{0,\Lambda_{B}^2 \right\}}
\no\\
&\ge& \mathbb{E}\left[ - J_{A}\tanh(\beta J_{A}) \right]_{\left\{0,\Lambda_{B}^2 \right\}} .
\ee
Furthermore, conducting integration by parts, we obtain Eq. (\ref{corr-ineq}).
}
\qed\end{proof}{}
%%%%%%%%%%%%%%%%%%%%%%%%%%%%%%%%%%%%%
A similar calculation is possible for higher order terms.
Taking $f(J_{A})=J_{A}^2$ in Eq. (\ref{inequality-1}), we obtain
\be
 \mathbb{E}\left[  -(J_{A}-J_{0,A})^3 \langle\sigma_A \rangle_{\{J_{B}\}}  \right]_{\left\{J_{0,B},\Lambda_{B}^2 \right\}}&=&\mathbb{E}\left[ -J_{A}^3  \langle\sigma_A \rangle_{\{J_{B}+J_{0,B}\}}  \right]_{\left\{0,\Lambda_{B}^2 \right\}}
\no\\
&\ge& \mathbb{E}\left[  -J_{A}^3 \tanh(\beta J_{A}) \right]_{\left\{0,\Lambda_{B}^2 \right\}} .
\ee
Conducting integration by parts and using Eq. (\ref{corr-ineq}), we obtain the lower bound on the expectation of the fourth power of the correlation function,
\be
 \mathbb{E}\left[  \frac{1}{6\beta^2\Lambda_{A}^2} \left(8\beta^2\Lambda_{A}^2 -3 \right) \left(\langle\sigma_A \rangle_{\{J_{B}+J_{0,B} \}}^2-\tanh^2(\beta J_{A}) \right) + \tanh^4(\beta J_{A}) \right]_{\left\{0,\Lambda_{A}^2 \right\}} 
 \le  \mathbb{E}\left[ \langle\sigma_A \rangle_{\{J_{B} \}}^4  \right]_{\left\{J_{0,B},\Lambda_{B}^2 \right\}}. \label{4th-corr-ineq}
\ee
For $8\beta^2\Lambda_{B}^2 \ge 3$, Eqs. (\ref{corr-ineq}) and (\ref{4th-corr-ineq}) yield
\be
 \mathbb{E}\left[ \tanh^4(\beta J_{A}) \right]_{\left\{0,\Lambda_{A}^2 \right\}} 
 \le  \mathbb{E}\left[ \langle\sigma_A \rangle_{\{J_{B} \}}^4  \right]_{\left\{J_{0,B},\Lambda_{B}^2 \right\}}.
\ee
Thus, for a sufficiently high temperature, the quenched average of the fourth power of the correlation function has a non-zero lower bound.

%%%%%%%%%%%%%%%%%%%%%%%%%%%%%%%%%%%%%%%%%%%%%%%%%%%%
%%%%%%%%%%%%%%%%%%%%%%%%%%%%%%%%%%%%%%%%%%%%%%%%%%%%
\subsection{Lower bound on spin-glass order-parameter in  Gaussian random-field Ising model}
Finally, we demonstrate that the spin-glass order-parameter in the Ising model in a Gaussian random field generally adopts a finite value at any temperature, regardless of the forms of the other interactions.

We consider the case where a random field, $\{h_i\},$ is independently applied to all the sites, where $\{h_i\}$  follows a Gaussian distribution with mean $J_{0}$ and variance $\Lambda_{}^2$.
The Hamiltonian is obtained as
\be
H&=&- \sum_{B \subset{V}} J_{B} \sigma_{B }
\no\\
&=&- \sum_{B\subset{V}\setminus{\{h_i\}}} J_{B} \sigma_{B}- \sum_{i=1}^N h_i \sigma_i  , \label{random-field-system}
\ee
where interaction $J_B$ other than $\{h_i\}$ takes any form.
Then, Eq. (\ref{corr-ineq}) is reduced to
\be
 \mathbb{E}\left[  \tanh^2(\beta h_i) \right]_{\left\{0,\Lambda^2 \right\}} \le  \mathbb{E}\left[ \langle\sigma_i \rangle_{\{J_{B} \}}^2  \right]_{\left\{J_{0},\Lambda_{}^2 \right\}} ,
\ee
which suggests that the quenched average of the square of the local magnetization has a non-zero value.

Furthermore, because the same inequality holds for all the sites, we obtain the following result:
%%%%%%%%%%%%%%%%%%%%%%%%%%%%%%%%%%%%%
\begin{corollary}
For the spin-glass order-parameter, $q$, 
\be
q=\frac{1}{N} \sum_i \mathbb{E}\left[ \langle\sigma_i \rangle_{\{J_{B} \}}^2  \right]_{\left\{J_{0},\Lambda_{}^2 \right\}},
\ee
the system (\ref{random-field-system}) satisfies the following inequality:
\be
 \mathbb{E}\left[  \tanh^2(\beta h_i) \right]_{\left\{0,\Lambda^2 \right\}} \le  q. \label{sg-order-finite}
\ee
\end{corollary}
%%%%%%%%%%%%%%%%%%%%%%%%%%%%%%%%%%%%%
Thus, when a Gaussian random field is applied, the spin-glass order-parameter generally has a non-zero lower bound.
In ferromagnetic models, the ferromagnetic order parameter, i.e., magnetization, has a finite value when a magnetic field is applied.
Equation (\ref{sg-order-finite}) suggests that a similar phenomenon occurs in the Ising model in a Gaussian random field.
This is a natural consequence; however, the existence of a finite lower bound is not obvious.

In addition, we note that Eq. (\ref{sg-order-finite}) does not indicate that there is a spin-glass phase in the Ising model in a Gaussian random field.

%%%%%%%%%%%%%%%%%%%%%%%%%%%%%%%%%%%%%%%%%%%%%%%%%%%%%%%%%%%%%%%%%%%%%%%%%%%
\section{ Conclusions}
In this study, we have obtained the lower bound on the local energy of the Ising model with quenched randomness.
We emphasize that the acquired inequality (\ref{inequality-2}) is independent of the other interactions.
Our result is a natural generalization of Ref. \cite{KNA} in which a symmetric distribution was considered.

Applying the obtained inequality to a Gaussian spin-glass model, we determine that the expectation of the square of the correlation function generally has a finite lower bound at any temperature.
Thus, the spin-glass order-parameter in the Ising model in a Gaussian random field generally adopts a finite value at any temperature, which is a natural but not a obvious result.

It is an interesting question whether a similar inequality as Eq. (\ref{corr-ineq}) will hold for a general distribution function of the random interactions.
Our proof relies on the property of a Gaussian distribution, and we have not obtained proofs for other distributions.

%%%%%%%%%%%%%%%%%%%%%%%%%%%%%%%%%%%%%%%%%%%%%%%%%%%%%%%%%%%%%%%%%%%%%%%%%%%%
\begin{acknowledgment}
The authors thank Shuntaro Okada for useful discussions.
The present work was financially supported by JSPS KAKENHI Grant Nos. 18H03303, 19H01095, and 19K23418, and JST-CREST grant (No. JPMJCR1402) of the Japan Science and Technology Agency.
\end{acknowledgment}

%%%%%%%%%%%%%%%%%%%%%%%%%%%%%%%%%%%%%%%%%%%%%%%%%%%%%%%%%%%%%%%%%%%%%%%%%%%%
%%%%%%%%%%%%%%%%%%%%%%%%%%%%%%%%%%%%%%%%%%%%%%%%%%%%%%%%%%%%%%%%%%%%%%%%%%%%

\end{document}